\documentclass{article}
\usepackage[utf8]{inputenc}
\usepackage[T1]{fontenc}
\usepackage[english]{babel}
\usepackage{amsmath,amsthm,amssymb,amsfonts}
\usepackage{xspace}
\usepackage{hyperref}
\usepackage{cleveref}
\usepackage{verbatim}
\usepackage{thm-restate}
\usepackage{tikz,graphicx}
\usepackage{caption}
\usepackage{subcaption}
\usepackage[hmargin=2.5cm,vmargin=3.8cm]{geometry}
\usepackage[textsize=footnotesize,color=green!40]{todonotes}

\usepackage{bm}

\newcommand{\CR}{\textsc{Cops and Robber}\xspace}

\newcommand{\ACR}{\textsc{Fully Active Cops and Robber}\xspace}

\newcommand{\R}{ \mathcal{R}}
\newcommand{\C}{ \mathcal{C}}
\newtheorem{theorem}{Theorem}
\crefname{theorem}{theorem}{theorems}
\newtheorem{lemma}[theorem]{Lemma}
\crefname{lemma}{lemma}{lemmas}
\newtheorem{proposition}[theorem]{Proposition}

\crefname{proposition}{proposition}{propositions}
\crefname{result}{result}{results}
\newtheorem{corollary}[theorem]{Corollary}
\crefname{corollary}{corollary}{corollaries}

\crefname{fact}{fact}{facts}
\newtheorem{observation}[theorem]{Observation}
\crefname{observation}{observation}{observations}
\newtheorem{question}[theorem]{Question}
\crefname{question}{question}{questions}

\crefname{claim}{claim}{claims}

\crefname{note}{note}{notes}

\crefname{conj}{conjecture}{conjectures}

\crefname{definition}{definition}{definitions}

\crefname{remark}{remark}{remarks}

\tikzstyle{noeud}=[circle,inner sep=2, minimum size =3 pt, line width = 1pt, draw=black, fill=white]

\title{On the Cop Number of String Graphs\thanks{A preliminary version of this paper appeared in ISAAC 2022~\cite{ourString}}}
\author{
Sandip Das\footnote{Indian Statistical Institute, Kolkata, India}
\and 
\and Harmender Gahlawat\footnote{G-SCOP, Grenoble-INP, France}
}

\begin{document}

\maketitle

\begin{abstract}
\textsc{Cops and Robber} is a well-studied two-player pursuit-evasion game played on a graph, where a group of cops tries to capture the robber. The \emph{cop number} of a graph is the minimum number of cops required to capture the robber. Gavenčiak et al.~[Eur. J. of Comb. 72, 45--69 (2018)] studied the game on intersection graphs and established that the cop number for the class of string graphs is at most 15, and asked as an open question to improve this bound for string graphs and subclasses of string graphs. We address this question and establish that the cop number of a string graph is at most 13. To this end, we develop a novel \textit{guarding} technique. We further establish that this technique can be useful for other \textsc{Cops and Robber} games on graphs admitting a representation. In particular, we show that four cops have a winning strategy for a variant of \textsc{Cops and Robber}, named \textsc{Fully Active Cops and Robber}, on planar graphs, addressing an open question of Gromovikov et al.~[Austr. J. Comb. 76(2), 248--265 (2020)]. In passing, we also improve the known bounds on the cop number of boxicity 2 graphs. Finally, as a corollary of our result on the cop number of string graphs, we establish that the chromatic number of string graphs with girth at least 5 is at most $14$. 
\end{abstract}

\section{Introduction}
\CR is a two-player perfect information pursuit-evasion game played on a graph. 
One player is referred as \textit{cop player}, who controls a set of \textit{cops}, and the other player is referred as \textit{robber player} and controls a single \textit{robber}. 
The game starts with the cop player placing each cop on some vertex of the graph, and multiple cops may simultaneously occupy the same vertex. Then, the robber player places the robber on a vertex of the graph. 
Afterwards, the cop player and the robber player make alternate moves, starting with the cop player. 
In a cop player move, each cop either moves to an adjacent vertex (along an edge) or stays on the same vertex. In the robber player move, the robber does the same. For simplicity, we will say that the cop (resp. robber) moves in a cop (resp. robber) move instead of saying that the cop (resp. robber) player moves the cop (resp. robber).

A state in the game where one of the cops occupies the same vertex as the robber is called the \textit{capture}. 
If the cops can capture the robber in a graph, then the cops \textit{win}, and if the robber can evade the capture forever, then the robber \textit{wins}. The \textit{cop number} of a graph $G$, denoted as $\mathsf{c}(G)$, is the minimum number of cops that can ensure the capture against all the strategies of the robber. For a family $\mathcal{F}$ of  graphs, $\mathsf{c}(\mathcal{F}) = \max\{ \mathsf{c}(G)~|~G \in \mathcal{F}\}$. In this paper, we consider finite, connected\footnote{The cop number of a disconnected graph is the sum of the cop numbers of its components; hence we assume connectedness.}, and simple undirected graphs. We denote the robber by $\R$.

The game of \CR was independently introduced by Quilliot~\cite{qui} and by Nowakowski and Winkler~\cite{nowakowski}, both in 1983, with just one cop. Aigner and Fromme~\cite{aigner} generalized the game to multiple cops and defined the cop number for a graph. 
The notion of cop number and some fundamental techniques introduced by Aigner and Fromme~\cite{aigner} have resulted in a plethora of rich results on this topic. For more details, we refer the reader to the book by Bonato and Nowakowski~\cite{bonatobook}.

The computational complexity of finding the cop number of a graph is a challenging question in itself. On the positive side, Berarducci and Intrigila~\cite{berarducci} provided a backtracking algorithm that decides whether the cop number of a graph is at most $k$ in $O(n^{2k+1})$ time; hence,  this is a polynomial-time algorithm for a fixed $k$.  
On the negative side, Fomin et al.~\cite{fomin} proved that determining the cop number of a graph is NP-hard as well as W[2]-hard. Moreover, the game was shown to be PSPACE-hard by Mamino~\cite{mamino} and EXPTIME-complete by Kinnersley~\cite{kinnersley}. Recently, Brandt et al.~\cite{ETHBound} provided the fine-grained lower bounds, and proved that the time complexity of any algorithm for \CR is $\Omega(n^{k-o(1)})$ conditioned on SETH, and $2^{\Omega (\sqrt{n})}$ conditioned on ETH.

A \textit{string representation} of a graph is a collection of simple curves on the plane such that each curve corresponds to a vertex of the graph, and two curves intersect if and only if the vertices they represent are adjacent in the graph. The graphs that have string representations are called \textit{string graphs}. Many important graph families like \textit{planar graphs}, \textit{chordal graph}, and \textit{disk graphs} are subfamilies of string graphs~\cite{dib1,dib11}. Pach and Toth~\cite{pachToth} proved that the number of string graphs on $n$ labeled vertices is at least $2^{\frac{3}{4} {n \choose 2}}$, arguing that many graphs are string graphs. \CR is well-studied on graphs having a representation either on the plane or on some surface of higher genus~\cite{aigner, quitor, lehner, schroeder}. Further, Andreae~\cite{andreae} established that in general for any graph $H$, at most $|E(H)|$ cops suffice to capture the robber on any graph which does not contain $H$ as a minor.  Gaven\v{c}iak et al.~\cite{gavenciak} studied the cop number of various families of intersection graphs and showed that the cop number for the class of string graphs is at most 15. We improve this result by giving a winning strategy using 13 cops for any string graph.

Several variations of \CR have been studied, and they vary mainly depending on the capabilities of the cops and the robber. Some of these variations are shown to have correspondence with various width measures of graphs like treewidth~\cite{seymour}, pathwidth~\cite{parsons1}, tree-depth~\cite{depth}, hypertree-width~\cite{adler}, cycle-rank\cite{depth}, and directed tree-width~\cite{dtwidth}. Moreover, Abraham et al.~\cite{cop-decs} defined  \textit{cop-decomposition}, which is based on the cop strategy in \CR game on minor-free graphs provided by Andreae~\cite{andreae}, and showed that it has significant algorithmic applications in theory.

Gromovikov et al.~\cite{active} studied a variation of the game, called \ACR, where the cops, as well as the robber, are forced to move to an adjacent vertex on their respective turns. We say that a cop (or robber) is \textit{active} if it has to move to an adjacent vertex in its every turn. Similarly, we say that a cop (or robber) is \textit{flexible} if, in its turn, it can either move to an adjacent vertex or stay on the same vertex. In \ACR, the cops, as well as the robber, are active. The \textit{active cop number} of a graph $G$, denoted  $\mathsf{c_a}(G)$, is the minimum number of cops required to ensure capture in \ACR. Gromovikov et al.~\cite{active} studied this game on various graph classes and suggested determining the active cop number of planar graphs as an open question. We address this question and show that the active cop number for the class of planar graphs is at most four. We also consider a variation of this game where only the cops are forced to be active. Let $\mathsf{c_A}(G)$ be the minimum number of active cops required to ensure the capture of a flexible robber in $G$. Observe that, for any graph $G$, $\mathsf{c_a}(G) \leq \mathsf{c_A}(G)$. We rather show that for a planar graph $G$, $\mathsf{c_A}(G) \leq 4$

\subsection{Preliminaries}
Let $u$ be a vertex of graph $G$. We define the \textit{open neighbourhood} of $u$, denoted by $N(u)$, as $\{v:uv\in E(G)\}$. We define the \textit{closed neighbourhood} of $u$, denoted by $N[u]$, as $N(u)\cup \{u\}$. 
For a subgraph $H$ of $G$, define the \textit{closed neighbourhood} of $H$, denoted by $N[H]$, as $\bigcup_{v\in V(H)} N[v]$. We also define $G-H$ as the graph induced by the vertices that are in $G$ but not in $H$. For a vertex $x \in V(G)$, we define $G-x$ as the graph induced by vertices in $V(G) \setminus \{x\}$.

Consider two arbitrary vertices $u,v \in V(G)$. By $d(u,v)$, we denote the distance between vertices $u$ and $v$ in $G$.
Let $P$ be a path of $G$ with endpoints $u$ and $v$. We say that $P$ is a $u,v$-path. Path $P$ is said to be \textit{isometric} if $P$ is a shortest $u,v$-path. Moreover, path $P$ is said to be \textit{convex} if every $u,v$-path $Q\neq P$ is longer than $P$. We remark that not every pair of vertices is guaranteed to have a convex path between them.
In this article, we do a lot of index manipulations on path vertices. Therefore, whenever we mention an isometric $v_i,v_j$-path $P$, for $i<j$, then the path is of the form $v_i,v_{i+1}, \ldots, v_{j-1},v_j$.

Let $H$ be a subgraph of $G$. We say that some cops are \textit{guarding} $H$ if $\R$ cannot enter a vertex of $H$ without getting captured. Similarly, we define a subset $T\subseteq V(G)$ as the \textit{robber territory} if $\R$ cannot move to a vertex $v \notin T$, without getting captured in the next move. 

Let $T \subseteq V(G)$ be the robber territory. A $u_0,u_k$-path $P$ is \textit{isometric} \textit{relative} to $T$, if there is no shorter $u_0,u_k$-path containing at least one vertex of $T$. An isometric $u_0,u_k$-path $P$ relative to $T$ is a \textit{convex path} \textit{relative} to $T$, if there exists no vertex  $x \in T$ such that $d(u_0,x) = i-1$ and $x\in N(u_i)$.


$G$ is an \textit{intersection graph} if each vertex $v\in V(G)$ corresponds to a set $\psi(v)$, and $uv \in E(G)$ if and only if $\psi(u) \cap \psi(v) \neq \emptyset$. 
A string graph $G$ is an intersection graph of \textit{strings}, where each string $\psi(v)$ is a continuous image of the interval $[0,1]$ into $\mathbb{R}^2$.  
Given a string graph $G$, we can generate strings corresponding to each vertex of $V$ such that two strings intersect if and only if the corresponding two vertices are adjacent in G. These strings are said to be a \textit{representation} of graph $G$. It is a standard assumption that for any string graph $G$, we can get a representation where the strings are non self-intersecting. So, we assume we have a representation where the strings are non self-intersecting.

\subsection{Our Results and Techniques Used}
It is well established that the cop number of a graph is well related to the geometry of the graph. This relation was first established by Aigner and Fromme~\cite{aigner}, who proved that the cop number of a planar graph is at most three. To show this, they proved the following result, which we will also use in this paper.

\begin{proposition}[\cite{aigner}]\label{res:shortest}
 Let $P$ be an isometric $u_0,u_k$-path in $G$. Then one cop can guard $P$ after at most $k$ cop moves.
\end{proposition}

A similar idea was used by Beveridge et al.~\cite{udg}, who showed that three cops can prevent the robber from crossing an isometric path in any unit disk graph, and using this proved that nine cops have a winning strategy for unit disk graphs. Later, Gaven\v{c}iak et al.~\cite{gavenciak}  proved that the cop number of string graphs is at most 15. To establish this, they proved the following result, which we will also use in this paper.
\begin{proposition}[\cite{gavenciak}]\label{lem:shortr}
Let $u$ and $v$ be two distinct vertices of $G$ and $P$ be an isometric $u,v$-path relative to the robber territory $T\subseteq V(G)$. Then five cops can guard $N[P]$ after at most $k$ cop moves.
\end{proposition}

At a very high level, the idea of the above strategies is the following. The cops play the game assuming a fixed representation of the graph. The cop player employs three teams of cops, where each team can prevent the robber from crossing an isometric path. The cops begin by using one team to guard an isometric path, say $P_1$. Next, the cop player finds another isometric path, say $P_2$, such that the endpoints of $P_1$ and $P_2$ are the same, and guard it using the second team. Now, $\R$ cannot cross the paths $P_1$ and $P_2$, and hence is restricted to one of the faces formed by the boundary $P_1 \cup P_2$ in the embedding. Now, we can delete the part of the graph not accessible to $\R$. In the remaining graph, the cop player finds another isometric path, say $P_3$, such that the endpoints of $P_3$ are the same as the endpoints of $P_1$ and $P_2$, and guard $P_3$ using the third team of the cops. This further restricts $\R$ to either in a face formed by the boundary $P_1\cup P_3$ or in a face formed by the boundary $P_2 \cup P_3$. In either case, we can free one team of cops and keep repeating this process, and in each iteration, the robber territory is strictly reduced. Since the robber territory is initially the graph $G$, which is finite, these three teams eventually capture the robber.

Our main observation is that if an isometric path $P$ is the \textit{convex path}, then in some cases, we can employ a smaller number of cops to prevent the robber from crossing $P$. More specifically, we have the following lemma, which we prove in Section~\ref{sec:4cops}.

\begin{restatable}{lemma}{ConvexString}\label{lem:uniqueIsometric}
Let $u_0$ and $u_k$ be two distinct vertices of $G$ and   $P$ be a convex $u_0,u_k$-path relative to the robber territory $T \subseteq V(G)$. Then four cops can guard $N[P]$, after at most $k$ cop moves.
\end{restatable}

Using Lemma~\ref{lem:uniqueIsometric}, we prove that four cops can prevent $\R$ from crossing a convex path in a string graph representation. Then, using some novel techniques, we give a strategy such that whenever two teams of cops are employed to guard two isometric paths, one of the teams is guarding a convex path. This directly gives a cop winning strategy using 14 cops for string graphs. We further use some techniques to improve this bound to 13 cops. We have the following result, which we prove in Section~\ref{sec:string}.

\begin{restatable}{theorem}{ThmString}\label{th:string}
If $G$ is a string graph, then $\mathsf{c}(G)\leq 13$.
\end{restatable}

Petr et al.~\cite{AlgoCR} gave an algorithm that, given a graph $G$, can decide in $\mathcal{O}(kn^{k+2})$ time if $\mathsf{c}(G) \leq k$. Therefore, for any graph family $\mathcal{F}$, if $\mathsf{c}(\mathcal{F}) \leq \ell$, where $\ell \in \mathbb{N}$, then for any graph $G \in \mathcal{F}$, $\mathsf{c}(G)$ can be computed in $\mathcal{O}(n^{\ell+2})$ time. Thus, we have the following corollary.

\begin{restatable}{corollary}{CorroRunning}
If $G$ is a string graph, then $\mathsf{c}(G)$ can be computed in $\mathcal{O}(n^{15})$ time.
\end{restatable}

Aigner and Fromme~\cite{aigner} also showed that for a graph $G$ with girth\footnote{The \textit{girth} of a graph $G$ is the length of a shortest cycle contained in $G$.} at least five and minimum degree $\delta(G)$, $\mathsf{c}(G) \geq \delta(G)$. Inspired by this, Gaven\v{c}iak et al.~\cite{gavenciak} established the following interesting relation between the cop number of a graph $G$, its degeneracy, and hence its chromatic number.

\begin{proposition}[\cite{gavenciak}]\label{P:chromatic}
Let $\mathcal{F}$ be a hereditary class of graphs such that $\mathsf{c}(\mathcal{F}) \leq k$, for $k \in \mathbb{N}$. Then, every graph $G \in \mathcal{F}$ with girth at least five is $k$-degenerate and therefore, $k+1$-colorable.
\end{proposition}

Using Proposition~\ref{P:chromatic}, they established that every string graph with girth at least five is 16-colorable. Although the results of Fox and Pach~\cite{FoxPach} imply that the chromatic number of girth five string graphs is bounded, their results do not mention an explicit numerical bound. Moreover, it is known that the chromatic number of string graphs with girth four is unbounded~\cite{stringUnbounded}. We also note here that the chromatic number of girth (at least) five 1-string graphs is at most six~\cite{KostochkaNesetril}, where 1-string graphs are the graphs with a string representation where two strings intersect at most once. We have the following corollary on the chromatic number of string graphs using Proposition~\ref{P:chromatic} and Theorem~\ref{th:string}.

\begin{restatable}{corollary}{CorroString}
is a string graph with girth at least five, then
\end{restatable}

Let \textit{2-BOX} be the family of intersection graphs of axis-parallel rectangles in $\mathbb{R}^2$. It is known that $2 \leq \mathsf{c}$\textit{(2-BOX)}$\leq 15$~\cite{gavenciak}. We improve this result in the following theorem.
\begin{restatable}{theorem}{ThmRectangle}\label{th:rectangle}
Let \textit{2-BOX} be the family of rectangle intersection graphs. Then $3 \leq \mathsf{c}$\textit{(2-BOX)}$\leq 13$
\end{restatable}
\begin{proof}
Since \textit{2-BOX} is a subclass of string graphs, the upper bound follows from Theorem~\ref{th:string}.
To prove the lower bound (i.e., $3 \leq \mathsf{c}$\textit{(2-BOX)}), we observe that the dodecahedron graph, having cop number three~\cite{aigner}, is a boxicity 2 graph. For completeness, we give a  rectangle intersection representation of the dodecahedron graph in Figure~\ref{fig:dode}.   
\end{proof}

\begin{figure}
    \centering
    \includegraphics[scale=0.25]{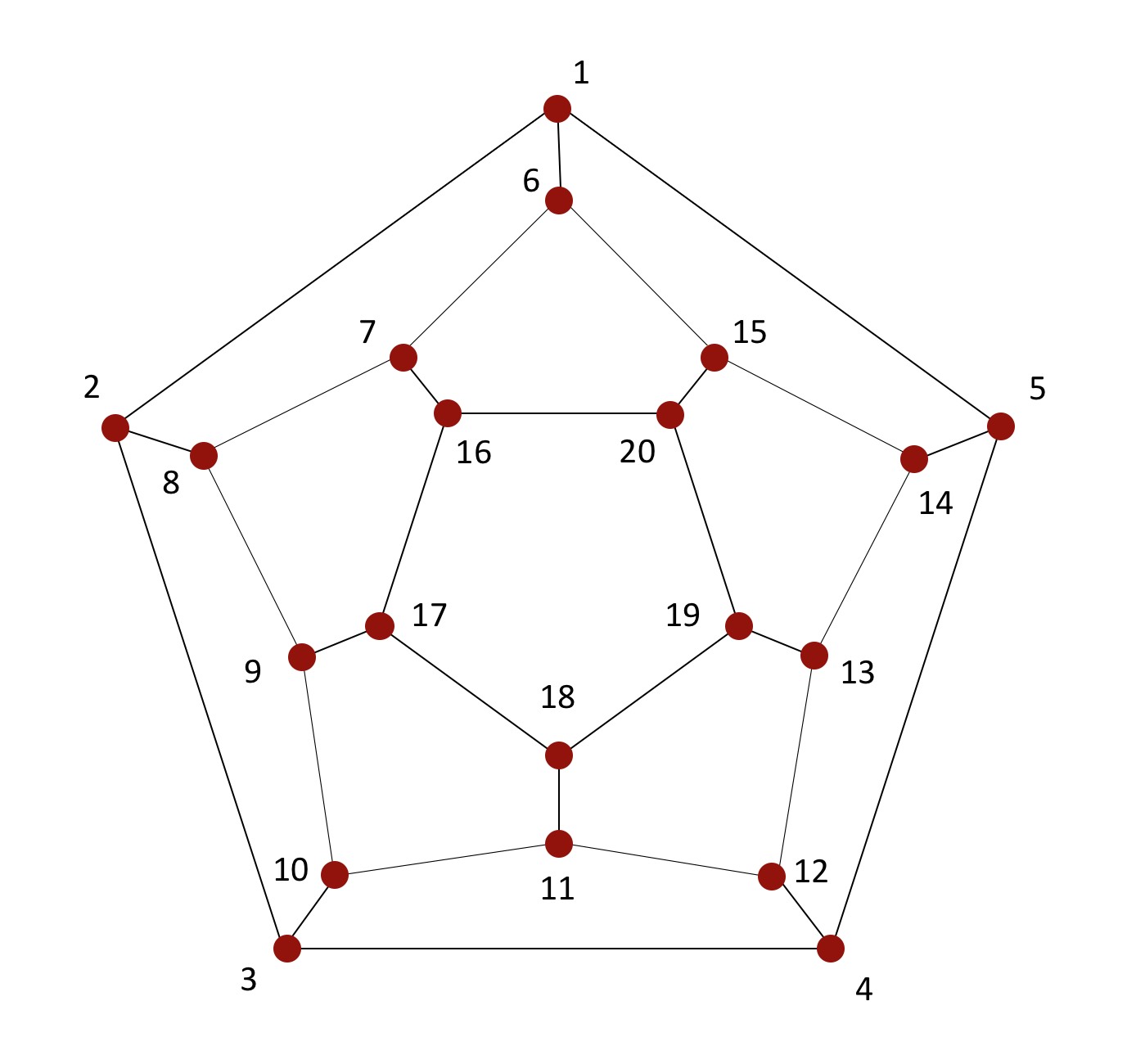}
    \includegraphics[scale =0.54]{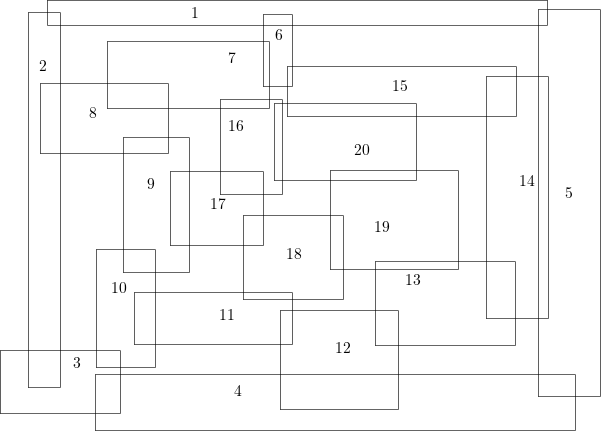}
    \caption{The dodecahedron and its boxicity 2 representation. Here each vertex $i$ corresponds to rectangle $i$.}
    \label{fig:dode}
\end{figure}
We further show that our technique can be used to attain better bounds for different variations of the game on other graph classes as well. In particular, we study \ACR on planar graphs.  It is known that $\mathsf{c_a}(G) \leq 2\cdot \mathsf{c}(G)$\cite{active}. Let $\mathcal{P}$ be the class of planar graphs, then trivially $\mathsf{c_a}(\mathcal{P}) \leq 6$. Gromovikov et al.~\cite{active} asked as open question what is the value of $\mathsf{c_a}(\mathcal{P})$. We answer this question partially by showing that $\mathsf{c_a}(\mathcal{P}) \leq 4$. To show this, we prove the following lemma in Section~\ref{sec:4cops}.

\begin{restatable}{lemma}{ConvexPlanar}\label{lem:uniquePlanar}
Let $v_0$ and $v_k$ be two distinct vertices of $G$ and $P$ be a convex $v_0,v_k$-path relative to the robber territory $T \subseteq V(G)$. Then, one active cop can guard $P$ against a flexible robber, after at most $k$ cop moves.
\end{restatable}
In Lemma~\ref{lem:uniquePlanar}, the active cop can guard a convex path even if the robber is flexible. Therefore, using Lemma~\ref{lem:uniquePlanar}, and techniques similar to the ones we use for string graph, we prove the following result in Section~\ref{sec:Planar}, a corollary of which is that $\mathsf{c_a}(G) \leq 4$.

\begin{restatable}{theorem}{ThmPlanar}~\label{th:planar}
If $G$ is a planar graph, then $\mathsf{c_A}(G) \leq 4$.
\end{restatable}

\section{Guarding Convex Paths}\label{sec:4cops}
In this section, we prove Lemma~\ref{lem:uniqueIsometric} and Lemma~\ref{lem:uniquePlanar}. We recall that, an isometric $u_0,u_k$-path $P$ is a convex path relative to the robber territory $T \subseteq V(G)$ if there exists no vertex  $x \in T$ such that $d(u_0,x) = i-1$ and $x\in N(u_i)$. Figure~\ref{fig:proofConvex} aids the proof of the following lemma.

\begin{figure}
    \centering
    \includegraphics[scale=0.9]{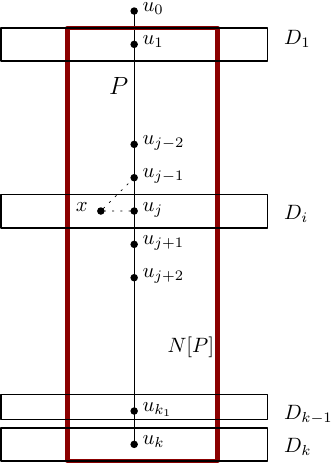}
    \caption{$P$ is a convex $u_0,u_k$ path (relative to $T$). Here $N[P]$ is denoted by heavier bold lines. The vertex $x\in (N[P]\cap D_i)\setminus\{u_i\}$. The only two possible edges between $x$ and vertices of $P$ are denoted by dotted lines.}
    \label{fig:proofConvex}
\end{figure}

\ConvexString*
\begin{proof}
We mark one cop as the \textit{sheriff}, and the other three cops are said to be its \textit{deputies}. The deputies follow the movements of the sheriff such that when the sheriff is at a vertex $u_j$, for $0 \leq j \leq k$, the deputies are at vertices $u_{j-2},u_{j-1}$ and $u_{j+1}$. Let the vertex $u_{k+1}$ refer to the vertex $u_k$, and let vertices $u_{-1}$ and $u_{-2}$ refer to the vertex $u_0$. Let $D_j = \{v~|~d(u_0,v) = j,\ \text{if}\ j<k;\ \text{and}\ d(u_0,v) \geq j,\ \text{if} \ j=k \}.$ See Figure~\ref{fig:proofConvex} for an illustration of the proof.

Since $P$ is an isometric path relative to $T$, the sheriff can guard $P$ in at most $k$ steps using Proposition~\ref{res:shortest}. Moreover, it is worth mentioning that the sheriff can do so by staying on the vertices of $P$. More specifically, after each move of the sheriff, if $\R$ is at a vertex $v \in D_j$, then the sheriff is at vertex $u_j$.
We claim that once the sheriff guards $P$, these four cops guard $N[P]$. 

To prove the above claim, we show that if $\R$ moves to a vertex $x \in N[P]$ (also $x \in T$), then $\R$ gets captured by one of the cops. If $\R$ moves to a vertex in $P$, then the sheriff will capture the robber as it is guarding $P$. Let  $\R$ moves to a vertex $x \notin V(P)$, $x \in N[P]$, and $x \in D_j$. 

Let $1 < j < k$. Since $x \in N[P]$, $x$ is adjacent to at least one vertex of $P$. Now $x$ cannot be adjacent to a vertex $y$ from $\{u_0, \ldots, u_{j-2} \}$, as through path $u_0, \ldots,y,x$ the distance $d(u_0,x) < j$, which is not possible since $x \in D_j$. Moreover, $x$ cannot be adjacent to a vertex $y$ from $\{ u_{j+2}, \ldots u_k \}$, as the path $u_0, \ldots x,y, \ldots u_k$ becomes a shorter $u_0,u_k$-path than  $P$, which is a contradiction to the fact that $P$ is an isometric path relative to $T$. Also, $x$ cannot be adjacent to $u_{j+1}$ by the definition of the convex path. Hence, $x$ can only be adjacent to $u_{j-1}$ and $u_j$, and is adjacent to at least one of them. Since the sheriff is guarding $P$, it can reach $u_j$ in this cop move, and hence is at one of the vertices from $\{u_{j-1},u_j,u_{j+1} \}$. In any case, there are cops on both $u_j$ and $u_{j-1}$. Therefore, one of these cops will capture $\R$ whenever $\R$ enters $x$.

Similar arguments hold for $j \in \{0,1,k \}$. If $j=k$, then observe that $x$ can only be adjacent to $u_k$ and $u_{k-1}$, and both these vertices would be occupied by cops. If $j= 1$, then $x$ can only be adjacent to $u_0$ and $u_1$, and both these vertices would be occupied by cops. If $j = 0$, then $x=u_0$ and hence $x$ in on $P$, and since the sheriff is guarding $P$, it will capture $\R$.  

Hence, these four cops can guard $N[P]$ in at most $k$ steps.
\end{proof}

Next, we show that for a convex path $P$ relative to the robber territory $T$, one active cop can guard $P$ against a flexible robber.

\ConvexPlanar*
\begin{proof}
Let $D_j = \{v~|~d(v_0,v) = j,\ \text{if}\ j<k;\ \text{and}\ d(v_0,v) \geq j,\ \text{if} \ j=k \}.$ We claim that if the cop $\C$ can ensure the following invariant, then $\C$ successfully guards $P$: \textit{after each move of the cop, if the robber is at a vertex in $D_i$, then $\C$ is at either $v_i$ or $v_{i-1}$}. Assume that this invariant holds and $\R$ moves to enter a vertex $v_r$ of $P$ from a vertex $u \notin V(P)$. Observe that, since $P$ is a convex path, $u \in D_r \cup D_{r+1}$. Consider the game state just before this move of $\R$. Due to the invariant condition, since $\R$ is at vertex $u$, the cop $\C$ is either at $v_r$, $v_{r-1}$, or $v_{r+1}$. In any of these case, $\C$ can move to capture $\R$ if $\R$ moves to $v_r$. 

Thus, if $\C$ can maintain this invariant, $\C$ guards $P$. Now, it remains to show that $\C$ can always reach this invariant and, once achieved, can always maintain it. $\C$ starts at vertex $v_0$. If  $\R$ is at vertex $u \in D_i$, then $\C$ assumes $image(\R)$ at vertex $v_i$. Since $image(\R)$ is restricted to $P$, $\C$ can capture $image(\R)$ in at most $k$ cop moves. Once $\C$ captures $image(\R)$, observe that we get the invariant. After that, $\C$ follows the following strategy: 
\begin{enumerate}
    \item If $\R$ moves from a vertex $u \in D_i$ to a vertex $w \in D_{i-1}$, then $\C$ moves to vertex $v_{j-1}$ from vertex $v_j$.
    
    \item  If $\R$ moves from a vertex $u \in D_i$ to a vertex $w \in D_{i+1}$, then $\C$ moves to vertex $v_{j+1}$ from vertex $v_j$.
    
    \item If $\R$ moves from a vertex $u \in D_i$ to a vertex $w \in D_{i}$: 
    \begin{enumerate}
        \item If $\C$ is at $v_i$, then it moves to $v_{i-1}$.
        \item If $\C$ is at $v_{i-1}$, then it moves to $v_i$.
    \end{enumerate}
\end{enumerate}
Since the above strategy maintains the invariant, this completes our proof.\end{proof}

\section{Cops and Robber on String Graphs} \label{sec:string}
\subsection{Definitions and Preliminaries} \label{sec:prelim}


\noindent\textbf{Segments, Faces and Regions:} A set $A \subset \mathbb{R}^2$ is \textit{arc-connected} if for any two points $a,b \in A$, the set $A$ contains a curve with endpoints $a$ and $b$. Consider a fixed string representation $\Psi$ of $G$. If two strings $\pi$ and $\pi'$ intersect at a point $p$, then we call $p$ as an \textit{intersection point}. In a fixed representation of a string graph $G$, a string can have multiple intersection points, and two strings can have multiple intersection points in common. A \textit{segment} $s$ of a string $\pi$ is a maximal continuous part of the string $\pi$ that does not contain any intersection point other than its endpoints. A string containing $k$ intersection points has $k+1$ segments.

A \textit{region} is an arc-connected area bounded by some segments of a set of strings in a string representation. A region also includes its boundary. Whenever we mention a region, it should satisfy our region definition. A \textit{face} is a region not containing any intersection point between two strings except on the boundary, and each continuous part of a string in the region intersects the boundary of the region at most once. It is a standard assumption that for a finite string graph $G$, we can have a representation such that the number of segments, intersection points, and faces is finite.

\begin{figure}
\centering
\begin{subfigure}{.5\textwidth}
  \centering
  \includegraphics[width=.92\linewidth]{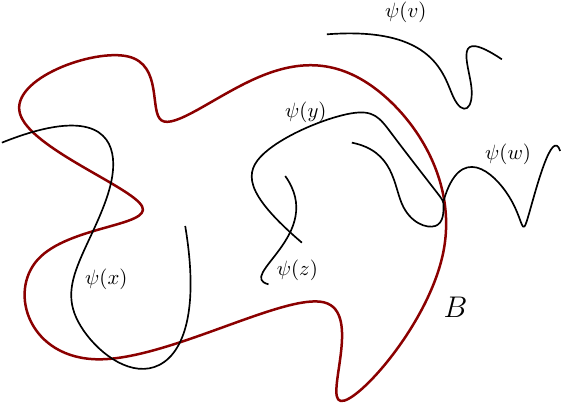}
  \caption{$\Psi$.}
  \label{fig:OC1}
\end{subfigure}%
\begin{subfigure}{.5\textwidth}
  \centering
  \includegraphics[width=.8\linewidth]{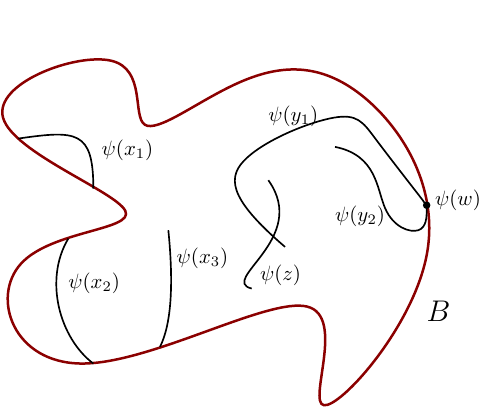}
  \caption{$\Psi_B$.}
  \label{fig:OC2}
\end{subfigure}
\caption{Here (a) represents $\Psi$ and (b) represents $\Psi_B$. $\psi(v)$ is not in $\Psi_B$, $\psi(z)$ is in $\Psi_B$. Further, for $\psi(x)$, strings $\psi(x_1)$, $\psi(x_2)$, and $\psi(x_3)$ are in $\Psi_B$; for $\psi(y)$, strings $\psi(y_1)$ and $\psi(y_2)$ are in $\Psi_B$; and for $\psi(w)$, the string $\psi(w_1)$, which is a single point, is in $\Psi_B$.} 
\label{fig:representation}
\end{figure}


\medskip
\noindent\textbf{Representation Restricted to a Region:} Consider a region $B$ of representation $\Psi$. We define the \textit{representation restricted to $B$}, denoted by $\Psi_B$, in the following manner. If a string $\psi(v)$ is completely inside $B$, then we have $\psi(v)$ in $\Psi_B$ also. If $\psi(v)$ is completely outside $B$, then $\psi(v)$ is not in $\Psi_B$. If a string $\psi(v)$ is such that some portion of $\psi(v)$ is outside $B$ and some portion of $\psi(v)$ is inside $B$, then we do the following. Let $s_1, \ldots, s_k$ be the portions of the string $\psi(v)$ such that each endpoint of $s_i$, for $0<i \leq k$, is either on the boundary of $B$ or is an endpoint of the string $\psi(v)$, and $s_i \in \Psi_B$. Then, instead of the string $\psi(v)$, we include $k$ new strings. We consider each portion $s_i$, for $0<i \leq k$, as a new string $\psi(v_i)$ in $\Psi_B$. See Figure~\ref{fig:representation} for an illustration. Let $G_B$ be the string graph corresponding to the representation $\Psi_B$. Here $V(G_B)$ is defined by the strings in $\Psi_B$, and $E(G_B)$ is defined by the intersections between these strings. 
Observe that, though $G_B$ might contain more vertices than $G$, the number of vertices in $G_B$ remains finite. Moreover, the number of faces and segments in $\Psi_B$ is not more than that in $\Psi$. Here, we also say that $G_B$ is the graph $G$ \textit{restricted} to region $B$. 

We want to note here that if there is a string completely contained in another, we can safely delete it without changing the cop number of the corresponding graph.  To see this observe that if a string $\psi(x)$ is contained in $\psi(y)$, then $N(x)\subseteq N(y)$, and hence deletion of $x$ does not change the cop number of the input graph (see Corollary~3.3 of~\cite{berarducci}). 

\medskip
\noindent\textbf{Relating Curves to Paths:} Let $\Psi$ be a fixed representation of a string graph $G$. Consider a curve $\pi$ in the representation $\Psi$.  $\pi$ is composed of some of the segments of the strings from $\Psi$. Let $\pi$ be composed of segments $s_1, \ldots, s_\ell$. Furthermore, consider a $u_1,u_k$-path $P$ in $G$. We say that the curve $\pi$ is \textit{related} to path $P$ if  each segment $s \in \{ s_1, \ldots, s_\ell\}$  is a segment of some string $\psi(u)$, $u \in \{u_1, u_2, \ldots, u_k\}$, and for each string $\psi(u)$, $u \in \{u_1, u_2, \ldots, u_k\}$, there is a segment $s \in \{s_1, \ldots, s_\ell\}$ such that $s$ is a segment of  $\psi(u)$.  Observe that $\ell\ge k$. 
Note that multiple curves may relate to the same path, and a curve may be related to multiple paths. For example, consider a complete graph $K_n$ (which is a string graph) and a string representation of $K_n$, denoted by $\Psi(K_n)$. If we choose a curve that contains at least one segment from each string of $\Psi(K_n)$, then this curve corresponds to every path of length $n-1$ in $K_n$. We would also like to mention here that the order of segments in the curve might not correspond to the order of vertices in the path. 

An \textit{isometric curve} in $\Psi$ is a curve that is related to an isometric path in $G$. We have the following observation that we use implicitly in our arguments.

\begin{observation}\label{Obs:5}
Although multiple isometric curves can be related to an isometric path, an isometric curve cannot be related to multiple isometric paths.
\end{observation}
\begin{proof}
Consider an isometric $u_1,u_k$-path $P$ and a curve $\pi$ related to $P$. Let $s_1,\ldots, s_l$ be the order of segments in $\pi$. Since $P$ is an isometric path, a segment of string $\psi(u_i)$ can only be adjacent to a segment of string $\psi(u_{i-1})$, $\psi(u_i)$, or of string $\psi(u_{i+1})$ in $\pi$. 

Let $z_1, \ldots, z_k$ be natural numbers such that $z_1 = 1$, $z_k = l$, and $z_1 < z_2 < \cdots < z_k$. Then there exists a sequence $z_1, \ldots, z_k$ such that each segment  $s \in \{ s_{z_i}, \ldots,s_{z_{i+1}} \}$, for $1 \leq i \leq k-2$, is a  segment of either the string $\psi(u_i)$ or the string $\psi(u_{i+1})$, and the segment $s_{z_{i+1}}$ is a segment of the string $\psi(u_{i+1})$. For $i=k-1$, each segment  $s \in \{ s_{z_i}, \ldots,s_{z_{i+1}} \}$, is a  segment of either the string $\psi(u_i)$ or the string $\psi(u_{i+1})$. Thus, $\pi$ can be related to only one path $ u_1, u_2, \ldots, u_k$. Therefore, an isometric curve relates to a unique isometric path.
\end{proof}

For ease of arguments in later proofs, we define \textit{monotone curves} in the following manner. Let $\pi$ be a curve related to an isometric $u_1,u_k$-path $P$, and let $s_1,\ldots, s_\ell$ be the order of segments of $\pi$. Let $z_1, \ldots, z_{k+1}$ be natural numbers such that $z_1 = 1$, $z_{k+1} = \ell$, and $z_1 < z_2 < \cdots < z_{k+1}$. Then, curve $\pi$ is said to be a \textit{monotone curve related} to $P$ if there exists a sequence $z_1, \ldots, z_{k+1}$ such that each segment  $s \in \{ s_{z_i}, \ldots,s_{z_{i+1}} \}$, for $1 \leq i \leq k-1$, is a  segment of  the string $\psi(u_i)$. For our results, whenever we consider an isometric curve related to an isometric path, we always consider a monotone curve, without mentioning it explicitly.

A curve with endpoints $a$ and $b$ is referred to as an $a,b$-\textit{curve}. Two curves are said to be \textit{internally disjoint} if they can intersect only at their respective endpoints. Let $\pi$ be a curve in a fixed representation $\Psi$. A curve $\pi'$ is said to be a \textit{sub-curve of} $\pi$ if $\pi'$ can be formed by some segments of $\pi$. We borrow the following topological lemmas by Gaven\v{c}iak et al.~\cite{gavenciak} that we will use.

\begin{lemma}[\cite{gavenciak}]\label{L:subpath1}
Let $B$ be a region. If $\pi$ is an isometric curve and $\pi' \subseteq \pi$ is a sub-curve with $\pi' \subseteq B$, then $\pi'$ is an isometric curve in  $\Psi_B$. 
\end{lemma}

\begin{lemma}[\cite{gavenciak}]\label{L:subpath2}
Let $\pi_1$ and $\pi_2$ be two internally disjoint isometric $a,b$-curves, with $a\neq b$, bounding a region $R$. For any simple $a,b$-curve $\pi_3 \subseteq R$ containing at least one interior point of $R$, we have that every arc-connected component of $R \backslash(\pi_1 \cup \pi_2 \cup \pi_3)$ is bounded by two simple and internally disjoint curves $\pi'_i$ and $\pi'_3$ with $\pi'_i \subseteq \pi_i$, $\pi'_3\subseteq \pi_3$ and $i \in \{1,2\}$.
\end{lemma}

If a curve $\pi$ is related to an isometric (resp. convex) path $P$ relative to $T$, then $\pi$ is referred to as an \textit{isometric} (resp. \textit{convex}) \textit{curve relative to} $T$.
We extend Lemma~\ref{L:subpath1} to accommodate the convex curves in the following lemma.
\begin{lemma}\label{subpath}
Let $B$ be a region of $\Psi$. If $\pi$ is a convex curve relative to $T$ and $\pi' \subseteq \pi$ is a sub-curve with $\pi' \subseteq B$, then $\pi'$ is a convex curve relative to $T$ in $\Psi_B$.  
\end{lemma}
\begin{proof}
First, we prove that $\pi'$ is a convex curve in $\Psi$ relative to $T$. Let the curve $\pi$ be related to the convex $u_0,u_k$-path $P$ in $G$. Then any sub-curve $\pi' \subseteq \pi$ would relate to a $u_i,u_j$-path $P' = u_i, u_{i+1}, \ldots, u_j$, where $0\leq i \leq j \leq k$. For contradiction, let us assume that $\pi'$ is not a convex curve relative to $T$ in $\Psi$, and hence $P'$ is not a convex path relative to $T$ in $G$. Thus, there exists a vertex $v \in (V(T) \setminus V(P))$ and some $u_\ell$ (where $i \leq \ell \leq j$) such that $d(u_i,v) = d(u_i, u_\ell)-1$ and $u_\ell \in N[v]$. Therefore, $d(u_0,u_i) + d(u_i,v)  = d(u_0,u_i)+ d(u_i, u_\ell)-1$. Hence, we have a vertex $v \in(V(T) \setminus V(P))$ such that $d(u_0,v) =d(u_0,u_\ell) - 1$ and $u_\ell \in N[v]$. This contradicts the fact that $P$ is a convex path related to the convex curve $\pi$, both relative to $T$. Hence, $\pi'$ is a convex curve in $\Psi$ and $P'$ is a convex path in $G$, both relative to $T$.

Next, we show that if a curve $\pi'$ is a convex curve in $\Psi$ relative to $T$ and $\pi' \subseteq \Psi_B$ (for some $B$ of $\Psi)$, then $\pi'$ is a convex curve in $\Psi_B$ relative to $T$. Consider two vertices $x$ and $y$ of $G$ corresponding to strings $\psi(x)$ and $\psi(y)$ in $\Psi$, respectively. Let $x'$ and $y'$ be two vertices in $G_B$ such that $\psi(x')$ is a portion of $\psi(x)$ and $\psi(y')$ is a portion of $\psi(y)$. Then observe that $d(x',y')$ in $G_B$ cannot be less than $d(x,y)$ in $G$. 

Now consider a vertex $v'$ in $G_B$ such that $v' \notin V(P')$, corresponding to string $\psi(v')$ in $\Psi_B$, such that $u_\ell \in N[v']$. Let $\psi(v)$ be a string in $\Psi$, corresponding to vertex $v$ such that $v \notin V(P)$, such that $\psi(v')$ is a portion of string $\psi(v)$ in $\Psi$. Hence $u_\ell$ is also a neighbour of $v$ in $G$. Since $P'$ is a convex path in $G$, either $d(u_i, v) = d(u_i,u_\ell)+1$ or $d(u_i, v) = d(u_i,u_\ell)$ in $G$. Hence, $d(u_i, v) \geq d(u_i,u_\ell)$ in $G$.  Since $d(x',y')$ in $G_B$ cannot be less than $d(x,y)$ in $G$, $d(u_i, v') \geq d(u_i,u_\ell)$. Thus, there cannot be any vertex $v'$ in $V(G_B) \setminus V(P')$ such that $u_\ell \in N[v']$ and  $d(u_i, v') = d(u_i,u_\ell)-1$. Hence, $P'$ is a convex path in $G_B$ and $\pi'$ is a convex curve in $\Psi_B$, both relative to $T$.
\end{proof}
We note here that if a path $P$ is an isometric (resp. convex) path relative to $T\subseteq V(G)$, then $P$ is an isometric (resp. convex) path relative to every subset $T' \subseteq T$. Similarly, if a curve $\pi$ is an isometric (resp. convex) curve relative to  $T$, then $\pi$ is an isometric (resp. convex) curve relative to every subset $T' \subseteq T$. Moreover, for a curve $\pi$ related to a path $P$, we say that $\pi$ is guarded if $N[P]$ is guarded.





\subsection{Bounding the Robber Region}

\noindent\textbf{Geometric Robber Territory:} Consider a fixed string representation $\Psi$ of a string graph $G$. We extend the definition of robber territory to the representation in the following manner. Consider a region $B$. Let $\R$ be on a vertex $u$ such that all points of the string $\psi(u)$ are inside the region $B$ and $\psi(u)$ does not intersect the boundary of $B$. Then we say that  $\Psi_B$ is the \textit{geometric robber territory} if $\R$ cannot move to a vertex $v$ such that the string $\psi(v)$ intersects the boundary of $B$, without getting captured. We note that although the strings that intersect the boundary of $B$ might be in $\Psi_B$, they are not accessible to $\R$ when we say that $\Psi_B$ is the geometric robber territory. See Figure~\ref{fig:geometric} for an illustration. Below, we show three ways to restrict the geometric robber territory.

\begin{figure}
    \centering
    \includegraphics[width=0.45\linewidth]{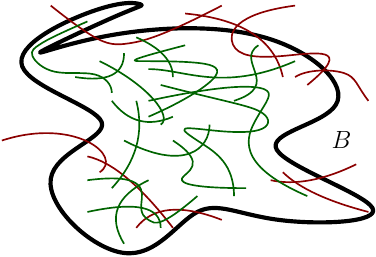}
    \caption{Illustration of geometric robber territory. Here the boundary of $B$ is depicted in bold black, the strings in $\Psi_B$ that do not intersect with boundary of $B$ are depicted in green, and the strings that intersect with the boundary of $B$ are depicted in red. $\Psi_B$ is the geometric robber territory if: (i) $\R$ is on a green string, and (ii) $\R$ gets captured if it moves to a red string. Although parts of red strings are in $\Psi_B$, none of these strings are accessible to $\R$.}
    \label{fig:geometric}
\end{figure}

Let $B$ be a region bounded by two internally disjoint $a,b$-curves $\pi_1$ and $\pi_2$. Then we also denote $\Psi_B$ by $\Psi_{\pi_1,\pi_2}$. Note that, here $B$ is the geometric robber territory if $\R$ is on a vertex $v \in V(G_B)\setminus (N[P_1] \cup N[P_2])$ and $\R$ cannot move to a vertex $u \in N[P_1 \cup P_2]$. Hence, we have the following observation. See Figure~\ref{fig:2Path} for an illustration.

\begin{observation}\label{O:curves}
Let $\pi_1$ and $\pi_2$ be two internally disjoint $a,b$-curves and $\R$ is in the region $\Psi_{\pi_1,\pi_2}$. If both curves $\pi_1$ and $\pi_2$ are guarded, then $\R$ cannot leave the region $\Psi_{\pi_1,\pi_2}$ without getting captured, and $\Psi_{\pi_1,\pi_2}$ becomes the geometric robber territory.
\end{observation}

\begin{figure}
\centering
\begin{subfigure}{.5\textwidth}
  \centering
  \includegraphics[width=.6\linewidth]{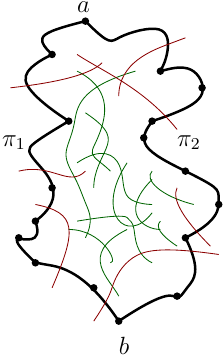}
  \caption{Geometric robber territory $\Psi_{\pi_1,\pi_2}$ defined by two internally disjoint $a,b$-curves $\pi_1$ and $\pi_2$.}
  \label{fig:2Path}
\end{subfigure}%
\begin{subfigure}{.5\textwidth}
  \centering
  \includegraphics[width=.6\linewidth]{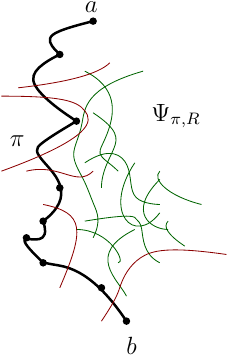}
  \caption{Geometric robber territory $\Psi_{\pi,R}$ defined by a single top-down $a,b$-curve $\pi$.}
  \label{fig:1Path}
\end{subfigure}
\caption{In both (a) and (b), the curves contain points to highlight the segments that form them. In both subfigures, $\R$ is on a green string and cannot access a red string without getting captured immediately.} 
\label{fig:Bound}
\end{figure}

Consider a fixed string representation $\Psi$ of a string graph $G$ in $\mathbb{R}^2$. We say that a string $\psi(v)$ is a \textit{top-most string} (\textit{bottom-most string}) if some point on $\psi(v)$ has the highest (lowest) $y$-$coordinate$ in $\Psi$. 
Here, we also say that $v$ is a \textit{top-most vertex} (\textit{bottom-most vertex}). 
Let $u$ and $v$ be two distinct vertices of $G$ such that $u$ is a top-most and $v$ is a bottom-most vertex. 
Let $a$ and $b$ be points on strings $\psi(u)$ and $\psi(v)$, respectively, such that $a$ and $b$ has the highest and lowest $y$-$coordinate$ in $\Psi$, respectively. Then, an $a,b$-curve $\pi$, related to an isometric $u,v$-path $P$, is referred to as a \textit{top-bottom} curve. Note that this curve may not be unique. 

Observe that if a vertex $x \notin N[P]$, then $\psi(x)$ lies either completely on the left of $\pi$ or completely on the right of $\pi$, and $\psi(x)$ does not intersect with $\pi$. If a string $\psi(x)$ lies on the left (or right) of the curve $\pi$, then we also say that vertex $x$ lies on the left (or right) of $P$. We say that the robber \textit{crosses} the curve $\pi$ (or the path $P$) if $\R$ moves (in some finite rounds) from a vertex $u$, completely on the left of $\pi$, to a vertex $v$, completely on the right of $\pi$, or vice versa.

We extend this idea of bounding a region $B$ with two internally disjoint curves $\pi_1$ and $\pi_2$, to bounding the region on the \textit{left} or the \textit{right} of a top-bottom curve $\pi$.  The region $B$ on the left (right) of the curve $\pi$ contains the points both on $\pi$ and on the left (right) of $\pi$. Here $\Psi_B$ is defined analogously and is also denoted by $\Psi_{\pi,L}$ ($\Psi_{\pi,R}$). We have the following observation. See Figure~\ref{fig:1Path} for an illustration.

\begin{observation}\label{O:path}
Let $\pi$ be a top-bottom curve related to an isometric path $P$. Then five cops can restrict the geometric robber territory to either $\Psi_{\pi,L}$ or $\Psi_{\pi,R}$.
\end{observation}
\begin{proof}
The curve $\pi$ is a continuous curve from a top-most point $a$ to a bottom-most point $b$ in the string representation. Hence, if a vertex $x$ is on left of $P$ and a vertex $y$ is on right of $P$, every path from $x$ to $y$ passes through a vertex of $N[P]$. Now, five cops can guard $N[P]$ using Proposition~\ref{lem:shortr}. This restricts $\R$ to cross the curve $\pi$. Hence, if $\R$ was on a vertex $x$ such that $\psi(x)$ is on the left (right) of $\pi$, then $\Psi_{\pi,L}$ ($\Psi_{\pi,R}$) becomes the geometric robber territory.
\end{proof}


The following observation provides one more way to bound the robber territory. 

\begin{observation}\label{O:vertex}
Let $x$ be a cut vertex of $G$ such that $G-x$ gives a connected component $G'$. If $\mathcal{R}$ is on a vertex  $v \in V(G')$ and a cop is occupying the vertex $x$, then $\R$ is restricted to $V(G')$ and $V(G')$ becomes the robber territory.
\end{observation}

We also define isometric (resp. convex) paths relative to geometric robber territories. We say that a path $P$ is \textit{isometric} (resp. \textit{convex}) \textit{path relative} to $\Psi_B = \Psi_{\pi_1,\pi_2}$ if $P$ is an isometric (resp. convex) path relative to $T = V(G_B) \setminus V(P_1 \cup P_2)$. Similarly, a curve $\pi$ is \textit{isometric} (resp. \textit{convex}) \textit{curve relative} to $\Psi_{\pi_1,\pi_2}$ if $\pi$ is an isometric (resp. convex) curve relative to $T = V(G_B) \setminus V(P_1 \cup P_2)$.

\subsection{Extending an Isometric Path/Curve}

Informally speaking, in this section, we show that if one team of cops is guarding an isometric path and we want to employ a new team of cops to guard another isometric path, then the new team can always find an isometric path such that in the region bounded by some specific curves of these paths, the first path is a convex path relative to the region bounded.  

\begin{figure}
    \centering
    \includegraphics[scale=1.2]{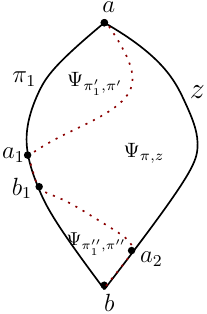}
    \caption{An illustration of extending an isometric curve. Here, $\Psi_{\pi_1,z}$ is defined by two $a,b$-curves $\pi_1$ and $z$ such that $\pi_1$ is an isometric curve relative to $\Psi_{\pi_1,z}$. Moreover, $\pi$ is denoted by dotted red curve, and $\pi$ overlaps with $\pi_1$ from point $a_1$ to $b_1$ and with $z$ from points $a_2$ to $b$. Let $a,a_1$-subcurve of $\pi$ (resp., $\pi_1$) be $\pi'$ (resp., $\pi'_1$) and $b_1,b$-subcurve of $\pi$ (resp., $\pi_1$) be $\pi''$ (resp., $\pi''_1$). If $\pi$ is an extended curve of $\pi_1$, then $\pi'_1$ and $\pi''_1$ are convex curves in $\Psi_{\pi'_1,\pi'}$ and $\Psi_{\pi''_1,\pi''}$, respectively, and $\pi$ is an isometric curve in $\Psi_{\pi,z}$.}
    \label{fig:extended}
\end{figure}

Consider a fixed representation $\Psi$ of a string graph $G$. Let $\Psi_{\pi_1,z}$, where $z \in \{\pi_2,L,R \}$, be the geometric robber territory.  See Figure~\ref{fig:extended} for an illustration. Let $\pi_1$ be an isometric curve relative to $\Psi_{\pi_1,z}$. Then we say that a curve $\pi \in \Psi_B$ is an \textit{extended curve} of $\pi_1$ if in the region $\Psi_{\pi',\pi'_1}$ bounded by any internally disjoint curves $\pi' \subseteq \pi$ and $\pi'_1 \subseteq \pi_1$, the curves $\pi_1'$ and $\pi'$ are a convex curve and an isometric curve relative to $\Psi_{\pi',\pi'_1}$, respectively; and in the region $\Psi_{\pi,z}$, the curve $\pi$ is an isometric curve relative to $\Psi_{\pi,z}$. We also say that the path $P$ related to $\pi$ is an extended path of the path $P_1$ related to $\pi_1$. We would like to note here that there might be multiple extended paths of an isometric path. In this section, we show that if 5 cops are guarding $N[P_1]$, then using at most 4 extra cops (total 9 cops), we can reduce the geometric robber territory to either $\Psi_{\pi',\pi'_1}$ or to $\Psi_{\pi,z}$. We have the following lemma.

\begin{lemma}\label{L:extendPath}
 Consider a fixed representation $\Psi$ of a string graph $G$. Let $\Psi_{\pi_1,z}$, where $z \in \{\pi_2,L,R \}$, be the geometric robber territory and let $\pi_1$ be an isometric curve relative to $\Psi_{\pi_1,z}$ guarded by 5 cops. If $\pi_1$ is not a convex curve in $\Psi_{\pi_1,z}$, then we can find an extended curve $\pi$ of $\pi_1$, and  restrict the geometric robber territory to either $\Psi_{\pi,z}$ or to $\Psi_{\pi',\pi'_1}$ where $\pi' \subseteq \pi$ and $\pi'_1 \subseteq \pi_1$, using at most 4 extra cops.
\end{lemma}
\begin{proof}
In the first part of this proof, we show how to find an extended curve of $\pi_1$, if it exists. Let $P_1$ and $P_2$ be the isometric paths related to the $a,b$-curves $\pi_1$ and $\pi_2$, respectively; and let $u_0$ and $u_k$ be the endpoints of $P_1$ and $P_2$. If $z\in \{L,R\}$, then let $P_2 = \phi$. If there is no $u_0,u_k$-path in $G_{\pi,x}$ other than $P_1$ and $P_2$, then we say that the curve $\pi$ cannot be extended. Note that here $\pi_1$ is a convex path in $\Psi_{\pi_1,x}$.

Let $G_B$ be the graph corresponding to the geometric robber territory $\Psi_B = \Psi_{\pi_1,x}$. If $P_1$ is a convex path relative to $\Psi_{\pi_1,x}$, then we find a shortest $u_0,u_k$-path $P$ in $G_B$ other than $P_1$ and $P_2$. Observe that $P$ is an isometric path relative to $\Psi_{\pi_1,x}$, and also an extended path of $P_1$. Here, we can simply free one cop from $P_1$ (since $P_1$ is convex path relative to $\Psi_{\pi_1,x}$) and use it along with 4 new cops to guard $N[P]$, and we are done. Thus, we can fix any $a,b$-curve $\pi \subseteq \Psi_B $ related to $P$ as an an extended curve of $\pi_1$.

If $P_1$ is not a convex path relative to $\Psi_{\pi_1,x}$, then we do the following. Find the least $i$ such that there is a vertex $x \in V(G_B) \setminus V(P_1 \cup P_2)$ with $d(u_0,x) = i-1$ and $u_i \in N(x)$. Now consider the string $\psi(u_i)$ in $\Psi_{\pi_1,x}$. Let $p_u$ be the intersection point of strings $\psi(u_i)$ and $\psi(u_{i-1})$ and let $p_b$ be the intersection point of strings $\psi(u_i)$ and $\psi(u_{i+1})$, in the curve $\pi_1$. We note here that these intersection points $p_u$ and $p_b$ are well defined since we are considering monotone curves. Next, we define sub-curves $\pi_u, \pi_p,$ and $\pi_b$ of the curve $\psi(u_i)$. Let $\pi_p = \pi_1 \cap \psi(u_i)$. Let $\pi_u$ be the maximal continuous sub-curve of $\psi(u_i)$ such that $\pi_u \cap \pi_1 = p_u$. Similarly, let $\pi_b$  be the maximal continuous sub-curve of $\psi(u_i)$ such that $\pi_b \cap \pi_1 = p_b$. We note here that one or both of $\pi_u$ and $\pi_b$ might be single points. 

Let $X \subseteq V(G_B) \setminus V(P_1 \cup P_2)$  such that $X = \{ x \ |\ d(u_0,x) = i-1\ \text{and}\ u_i \in N(x)$\}.
Let $X_{\epsilon} \subseteq X$, where $\epsilon \in \{u,p,b\}$, such that $X_{\epsilon} = \{ x \ | \ \psi(x) \cap \pi_{\epsilon}  \neq \phi \}$. Now we find a suitable vertex $v\in X$ in the following manner. See Figure~\ref{fig:extend} for an illustration.

\begin{figure}
\centering
\begin{subfigure}{.33\textwidth}
  \centering
  \includegraphics[width=.8\linewidth]{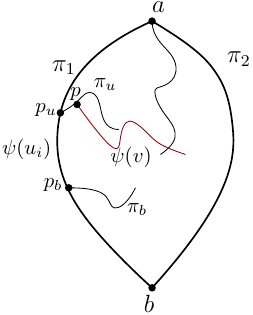}
  \caption{Case 1. Here $p'=p_u$.}
  \label{fig:extendPu}
\end{subfigure}%
\begin{subfigure}{.33\textwidth}
  \centering
  \includegraphics[width=.8\linewidth]{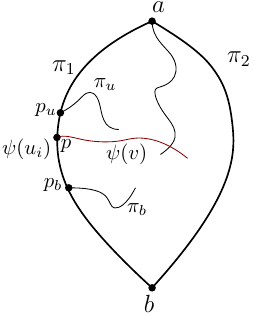}
  \caption{Case 2. Here $p'=p$.}
  \label{fig:extendP}
\end{subfigure}
\begin{subfigure}{.33\textwidth}
  \centering
  \includegraphics[width=.8\linewidth]{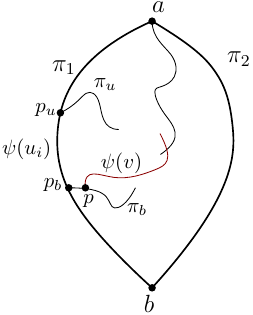}
  \caption{Case 3. Here $p'=p_b$.}
  \label{fig:extendPb}
\end{subfigure}
\caption{An illustration for extending an isometric curve. Here $\pi_1$ is an isometric curve relative to $\Psi_{\pi_1,\pi_2}$. The curve $\psi(v)$ is displayed in red and $p$ is an intersection point of $\psi(v)$ and $\psi(u_i)$.} 
\label{fig:extend}
\end{figure}

\begin{enumerate}
    \item If $X_u$ is not empty, then we find a vertex $v \in X_u$ such that an intersection point of the string $\psi(v)$ and curve $\pi_u$ is closest to the point $p_u$ along the curve $\pi_u$. We mark this intersection point as $p$ and we set $p' = p_u$. See Figure~\ref{fig:extendPu} for an illustration.
    
    \item If $X_u$ is empty and $X_p$ is not empty, then we find a vertex $v \in X_p$ such that an intersection point of the string $\psi(v)$ and curve $\pi_p$ is closest to the point $p_u$ along the curve $\pi_p$. We mark this intersection point as $p$ and we set $p' = p$. See Figure~\ref{fig:extendP} for an illustration.
    
    \item If $X_u$ and $X_p$ are empty, then we find a vertex $v \in X_b$ such that an intersection point of the string $\psi(v)$ and curve $\pi_b$ is closest to the point $p_b$ along the curve $\pi_b$. We mark this intersection point as $p$ and we set $p' = p_b$. See Figure~\ref{fig:extendPb} for an illustration.
\end{enumerate}

Now consider an isometric $u_0,v$-path $P_q$ relative to $\Psi_{\pi_1,x}$, and let $\pi_q$ be an isometric $a,p$-curve related to path $P'$. Moreover, let $\pi_r$ denote the $p,p'$-curve along $\psi(u_i)$, and $\pi_s$ denote the $p',b$-subcurve of the curve $\pi_1$ We compose the curve $\pi= \pi_q \cup \pi_r \cup \pi_s$. Observe that $\pi$ is an extended curve of $\pi_1$, and $P$ (path related to $\pi)$ is an extended path of $P_1$.  

In the next part of the proof, we show that if five cops are guarding the closed neighborhood of $P_1$ (with $\Psi_{\pi_1, \pi_2}$ being the geometric robber territory), then a total of nine cops can guard both $N[P]$ and $N[P_1]$ with $\Psi_{\pi_1, \pi_2}$ being the geometric robber territory. 

Recall that if $Q$ is a convex $w_0,w_{\ell}$-path, then four cops can guard $N[Q]$ by the sheriff guarding the path $Q$ and the deputies being at vertices $w_{j-2},w_{j-1}$ and $w_{j+1}$, when the sheriff is at vertex $w_j$ (due to Lemma~\ref{lem:uniqueIsometric} and its proof). Similarly, five cops can guard the closed neighbourhood of an isometric $w_0,w_{\ell}$-path $Q$, by the sheriff guarding the path $Q$ and the deputies being at vertices $w_{j-2},w_{j-1}, w_{j+1}$ and $w_{j+2}$, when the sheriff is at vertex $w_j$. Let us denote the deputy that moves to vertex $w_{j+2}$ when the sheriff moves at $w_j$ as the \textit{special deputy}.

Consider the paths $P$ and $P_1$. Note that both paths have the same length. Let the vertices of the path $P$ be denoted by $v_0, \ldots, v_k$. Note that $u_0 = v_0$ and for $\ell \geq i$, $u_{\ell} = v_{\ell}$. Now, we have that five cops are guarding $P_1$ (recall that $P_1$ is the isometric path corresponding to the isometric curve $\pi_1$). The main idea is that in the $u_0,u_i$-subpath of path $P_1$, only four cops are required (since it is a convex path), and if we keep the special deputy at $v_{j+2}$ when the sheriffs are at $u_j$ and $v_j$, then it serves as the special deputy for both paths $P$ and $P_1$ (because in path $P_1$, the special deputy is required only at vertices from $u_i, \ldots, u_k$, and note that these vertices are same as vertices $v_i, \ldots, v_k$). The extra four cops move on $P$ such that the sheriff guards $P$ and the deputies are at $v_{j-2}, v_{j-1}$, and $v_{j+1}$, when the sheriff is at $v_j$. When the sheriff successfully guards $P$, let it be at the vertex $v_{\ell}$. If ${\ell}\geq i$, then we have already achieved the goal as the special deputy is already at $v_{{\ell}+2}$ (since $v_{{\ell}+2} = u_{{\ell}+2}$ in this case). Otherwise ${\ell}<i$, and in this case, the special deputy is at $u_{{\ell}+2}$. Now, the special deputy starts moving towards $u_i$ irrespective of the moves of the sheriffs. Once it reaches $u_i$, it moves aiming to reach $v_{{\ell}+2}$ when the deputies are at $u_{\ell}$ and $v_{\ell}$. Once the special deputy reaches such a vertex, we have the desired state.

Now, if $\R$ is in the region $\Psi_{\pi,z}$, then we can free the cops from $\pi_1$ and $\R$ is restricted to $\Psi_{\pi,z}$. Otherwise $\R$ is restricted to the region $\Psi_{\pi',\pi'_1}$ where $\pi' \subseteq \pi$ and $\pi'_1 \subseteq \pi_1$.
\end{proof}

\subsection{Algorithm for String Graphs}
In this section, we show that for a string graph $G$, $\mathsf{c}(G) \leq 13$, by giving a winning strategy using 13 cops for any string graph. Let $\Psi$ be a fixed representation of $G$. First, we provide the intuition for our strategy. We define three "favorable" game states. Then we show that whenever we are in a favorable game state, 13 cops can force the game to another favorable game state such that the geometric robber territory gets reduced.

Let $\mathcal{R}$ be restricted to $\Psi_B$. Moreover, let $u$ and $v$ be two distinct vertices in $G_B$. For our strategy, first we define three game states, \textit{state 1}, \textit{state 2}, and \textit{state 3} as follows:
\begin{enumerate}
    \item \textit{State 1}: Let $u$ be a top-most and $v$ be a bottom-most vertex in $G_B$. Then five cops are guarding the closed neighbourhood of an isometric $u,v$-path $P$ in $G_B$. Note that, in $\Psi_B$ this restricts $\R$ to either $\Psi_{\pi,L}$ or $\Psi_{\pi,R}$, where $\pi$ is a curve related to $P$ (Observation~\ref{O:path}).
    
    \item \textit{State 2}: The region $B$ is bounded by two internally disjoint curves $\pi_1$ and $\pi_2$ such that $\pi_1$ is a convex curve in $\Psi_{\pi_1,\pi_2}$ and $\pi_2$ is an isometric curve in $\Psi_{\pi_1,\pi_2}$, both relative to $\Psi_{\pi_1,\pi_2}$. Then five cops are guarding $\pi_2$ and four cops are guarding $\pi_1$ (total 9 cops).
    
    \item \textit{State 3}: Let $x$ be a vertex in $G$ such that $G_B$ is a connected component of $G - x$. If a cop is occupying the vertex $x$ and $\R$ is in $G_B$, then observe that $\R$ is restricted to $G_B$ (Observation~\ref{O:vertex}). Let $u$ be a top-most vertex and $v$ be a bottom-most vertex in $G_B$, and $P$ be an isometric $u,v$-path relative to $T = V(G_B)$. Then one cop is occupying vertex $x$ and five cops are guarding $N[P]$. Moreover, $\mathcal{R}$ and $\psi(x)$ are on the same side of each curve $\pi$ related to $P$.
\end{enumerate}
State 1, state 2, and state 3 are referred to as the \textit{safe states}. We have the following lemma, which is central to our algorithm.
\begin{lemma}\label{lem:reduce}
Consider a fixed representation $\Psi$ of a string graph $G$. Let $\R$ be in a region $B$ of $\Psi$, and $\Psi_B$ be the geometric robber territory. Let the game be in a safe state $S$. Then 13 cops can force the game to a safe state $S'$ and the geometric robber territory to $\Psi_{B'} \subset \Psi_B$, in a finite number of moves.
\end{lemma}
\begin{proof}
Depending upon the state $S$ of the game, we do the following: 
\begin{enumerate}
    \item \textbf{S = state 1}: Let $u$ be a top-most and $v$ be a bottom-most vertex in $G_B$, and $P$ be the isometric $u,v$-path such that $N[P]$ is guarded by 5 cops. Let $\pi$ be a curve related to path $P$ and let $\pi$ defines $B$. Observe that the geometric robber territory is either $\Psi_{\pi,L}$ or $\Psi_{\pi,R}$. Without loss of generality, let us assume that $\R$ is restricted to the right of $\pi$ and hence $\Psi_{\pi,R}$ is the geometric robber territory. We extend the curve $\pi$ in $\Psi_{\pi,R}$ and let $\pi'$, related to a path $P'$, be an extended curve of $\pi$. Now, one of the following scenarios is possible:
    \begin{enumerate}
        \item Curve $\pi$ cannot be extended: It is possible only if there is no $u,v$-path in $G_B$ other than $P$. Let $\R$ be in a connected component  $G'$ of $G_B-P $. In this case, we claim that there is a unique vertex $x \in V(P)$ such that $x$ has a neighbor in $G'$. For contradiction, assume that there is some other vertex $y \neq x$ in $P$ such that $y$ has some neighbor in $G'$. Then consider the path $Q$ formed by the vertices of $u,x$-path along $P$, followed by a shortest $x,y$-path in $G'\cup \{u,v\}$, followed by the $y,v$-path along $P$.  Here $Q$ is a path other than $P$, and thus we have a contradiction. Thus $x$ is a cut vertex such that $G_B - x$ gives $G'$ as a component. 
        
       We guard $x$ using one cop and free other cops from $P$. Now, find a top-most vertex $u'$ and a bottom-most vertex $v'$ in $G'$ and an isometric $u',v'$-path in $G'$. Now, consider a top-down curve $\pi'$ corresponding to path $P'$ and guard $\pi'$ using five cops. If $\mathcal{R}$ and $x$ are on the same side of $P'$, then we are in the safe state 3. If $\mathcal{R}$ and $x$ are on opposite sides of $\pi'$, then we can free cop on $x$, and we are in the safe state 1. In both cases, at least the segments corresponding to the vertices of $V(P)-\{x\}$ will be removed from the geometric robber territory.
        
        \item $\mathcal{R}$ is on the same side of $\pi$ and $\pi'$: Since $\pi'$ is a top-bottom curve and $\pi'$ is guarded by five cops, we can free the cops on curve $\pi$. Hence, the geometric robber territory is now $\Psi_{\pi',R}$ (since $\R$ is in the right of both $\pi$ and $\pi'$). Also,  $\Psi_{\pi',R} \subset \Psi_{\pi,R}$ since the region bounded between $\pi$ and $\pi'$ is in $\Psi_{\pi,R}$ but not in $\Psi_{\pi',R}$.
        
        \item $\mathcal{R}$ is in the region bounded by two curves $\pi_1$ and $\pi_1'$ such that $\pi_1 \subseteq \pi$ and $\pi_1' \subseteq \pi'$: By the definition of extended curve, we know that $\pi_1$ is a convex curve and $\pi_1'$ is an  isometric curve, both relative $\Psi_{\pi_1,\pi_1'}$. Hence using Lemma~\ref{L:extendPath}, we can restrict the geometric robber territory to $\Psi_{\pi_1,\pi_1'}$ using at most 9 cops. Hence we are in the safe state 2. For the sake of simplicity, to prove that the geometric robber territory decreases in this case, we prove it for state 2, and whenever this case occurs, we execute this Lemma again for state 2.
    \end{enumerate}
    
    \item \textbf{S = state 2}: 
    Let $B$ be bounded by two internally disjoint $a,b$-curves $\pi$ and $\pi'$, and $\Psi_{\pi,\pi'}$ be the geometric robber territory. Let $\pi$ and $\pi'$ are related to $u,v$-paths $P$ and $P'$, respectively. Moreover, let $\pi$ be a convex curve in $ \Psi_{\pi,\pi'}$ and $\pi'$ be an isometric curve in $\Psi_{\pi,\pi'}$, both relative to $\Psi_{\pi,\pi'}$. Also, four cops are guarding $\pi$ and five cops are guarding $\pi'$.
    
    Now, if the curve $\pi'$ can be extended, then we extend the curve $\pi'$ using Lemma~\ref{L:extendPath}. Let $\pi_1$ be an extended curve of $\pi'$. Also, let $P_1$ be the path related to $\pi_1$ is an extended curve of $\pi'$. Now, using Lemma~\ref{L:extendPath}, we can guard both $\pi_1$ and $\pi'$ using at most nine cops. Now, if $\R$ is restricted in $\Psi_{\pi,\pi_1}$, then we can free the cops guarding $\pi'$, and we reach safe state 2. Note that the region bounded by curves $\pi'$ and $\pi_1$ is removed from the geometric robber territory. If $\R$ is restricted in $\Psi_{\pi'',\pi_1'}$ such that $\pi'' \subseteq \pi'$ and $\pi_1' \subseteq \pi_1$, then note that we can free the cops guarding $\pi$. Observe that the curve $\pi$ is removed from the geometric robber territory in this case.

    Suppose we cannot extend the curve $\pi'$ (that is, there is no $u,v$-path in $G_B$ other than $P$ and $P'$). Then observe that the vertices of the connected component of $G_B-(P\cup P')$ containing $\mathcal{R}$ can be connected to only one vertex $x$ of $P\cup P'$ (Proof is similar to the argument in case 1(a)). We move one cop to vertex $x$ and free all other cops. Now, we are in a situation similar to that of step 1(a). Hence we follow the same steps. Note that we also reduce the geometric territory of $\R$ in this step.

    \item \textbf{S = state 3}: Let $x$ be a vertex such that $G_B$ is a connected component of $G-x$. Consider the representation $\Psi' \subset \Psi$ such that $\Psi' = \{ \psi(u) ~|~  u\in G_B \}$.  
    Let $u$ and $v$ be a top-most and bottom-most vertex of $G_B$, respectively. Also, let $P$ be the isometric $u,v$-path such that $N[P]$ is guarded by five cops, and one cop is occupying the vertex $x$. 
    Moreover, both $\R$ and $x$ are on the same side of $P$. Without loss of generality, let us assume that they are on the right of $P$. Since $x$ is occupied by a cop and $N[P]$ is guarded by cops, observe that the geometric robber territory is $\Psi'_{\pi,R}$, where $\pi$ is a curve related to $P$. Now, if the curve $\pi$ can be extended, then we extend the curve $\pi$ in $\Psi'_{\pi,R}$ (using Lemma~\ref{L:extendPath}) and let  $\pi_1$ be the extended curve of $\pi$. Also let $P_1$ be the path related to $\pi_1$.
    
    If $\mathcal{R}$ and $x$ are on the same side of $\pi_1$, then we can free cops from $P$, and we are in the safe state 3. Here, the geometric robber territory is reduced by the region bounded between $\pi$ and $\pi_1$.
    
    If $\mathcal{R}$ is in the region bounded by two internally disjoint curves $\pi'$ and $\pi'_1$ such that $\pi' \subseteq \pi$ and $\pi'_1 \subseteq \pi_1$, then we are in the safe state 2 (by the definition of extended curves). Now, we free the cop guarding $x$. Here, the geometric robber territory is reduced by some segments of $\psi(x)$, at least.

    If the curve $\pi$ cannot be extended in $\Psi'_{\pi,R}$, then there exists a vertex $y \in V(P)$ such that vertices in $G_B-y$ gives a connected component $G_{B'}$ containing $\mathcal{R}$. If $x$ is not adjacent to any vertex in $V(G_{B'})$, then we are in a situation similar to 1(a), and we follow the same steps. If $x$ is adjacent to some vertex in $V(G_{B'})$, then We place one cop on $y$ and free other cops from $P$. Now, we find a top-most vertex $u'$ and bottom-most vertex $v'$ in $G_{B'}$ and find an isometric $u',v'$-path $P_1$ in $G_{B'}$. Now, five cops guard $N[P']$. Consider a top-bottom curve $\pi'$ related to $P'$. Now, either $x$ and $\mathcal{R}$ lie on the same side of $\pi'$ or  $y$ and $\mathcal{R}$ lie on the same side of $\pi'$. In both cases, we are in the safe state 3. Also, observe that each segment $s$ such that $s$ is a segment of path $P$ and $s$ is not a segment of string $\psi(y)$ is reduced from the geometric robber territory. Hence the geometric robber territory reduces in this step.
\end{enumerate}
This completes the proof of our lemma.
\end{proof}

Now we prove the main result of this section.
\ThmString*
\begin{proof}
We give a cop strategy to prove our claim. Consider a fixed representation $\Psi$ of a string graph $G$. We first show that at most 13 cops can force the robber to a safe state. Initially, let the robber territory be $\Psi$ and $G_B=G$. Cops find a top-most vertex $u$ and a bottom-most vertex $v$ in $G_B$ and find an isometric $u,v$-path $P$ in $G_B$. Now, five cops guard $N[P]$. This restricts the robber either to the left or to the right of $P$. Now we are in the safe state 1.

After this, until the robber is captured, we use Lemma~\ref{lem:reduce} to reduce the geometric robber territory. Since we have a finite graph with a finite representation and cops can reduce the geometric robber territory in every iteration of Lemma \ref{lem:reduce} using at most 13 cops, these 13 cops will eventually capture the robber. 
\end{proof}

\section{Fully Active Cops and Robber on Planar Graphs}\label{sec:Planar}
In this section, we show that for a planar graph $G$, $\mathsf{c_a}(G) \leq 4$.  We rather consider the game where only the cops are forced to be active and $\R$ is flexible. In this section, we show that for a planar graph $G$, $\mathsf{c_A}(G) \leq 4$. 
It is easy to see that for a graph $G$, $\mathsf{c}(G) \leq \mathsf{c_A}(G)$. Hence, there exist a planar graph $G$ such that $\mathsf{c_A}(G) \geq 3$. Here, we argue that four active cops have a winning strategy for any planar graph. First, we present a straightforward result.

\begin{lemma}\label{lem:guardPlanar}
Let $P$ be an isometric path of a  graph $G$. Then, two active cops can guard $P$, after a finite number of cop moves.
\end{lemma}
\begin{proof}
Let the two cops be denoted as \textit{sheriff} and \textit{deputy}. The two cops stay on adjacent vertices of $P$. The cops move such that the sheriff can guard $P$ using the strategy to guard $P$ in \CR setting using Proposition~\ref{res:shortest}. At some point during the game, if in the classical game strategy, the sheriff has to stay at a vertex on the cop move, the two cops switch positions and also switch the role of sheriff and deputy. This way, the cop that currently is the sheriff guards $P$.
\end{proof}

Using Lemma~\ref{lem:guardPlanar} and the strategy of Aigner and Fromme~\cite{aigner}, it is easy to see that $\mathsf{c_A}(G) \leq 6$. We use Lemma~\ref{lem:uniquePlanar} and Lemma~\ref{lem:guardPlanar} to improve this bound in the following theorem.

\ThmPlanar*
\begin{proof}
We can use Lemma~\ref{lem:uniquePlanar} and techniques similar to the techniques used for string graphs in Section~3 to show that four active cops have a winning strategy against a flexible robber.
\end{proof}

Since $\mathsf{c_a}(G) \leq \mathsf{c_A}(G)$, we have the following immediate corollary of Theorem~\ref{th:planar}.
\begin{corollary}
Let $P$ be a planar graph. Then $\mathsf{c_a}(G) \leq 4$.
\end{corollary}

\section{Final Remarks and Future Directions}\label{sec:final}
We proved that the cop number of a string graph is at most 13. But currently, we do not know any string graph having cop number at least four. Thus, for the class of string graphs $\mathcal{S}$, $3\leq \mathsf{c}(\mathcal{S}) \leq 13$. One immediate open question is to improve this bound by either giving a strategy for fewer cops or by giving an explicit construction of a string graph having cop number at least four. It might also be interesting to tighten the bounds on the active cop number of planar graphs.

\CR is also well-studied with regard to graph genus. Quillot~\cite{quitor} showed that for a graph $G$ having genus $g$, $\mathsf{c}(G) \leq 2g+3$. He used an ``unfolding'' technique where cops find and guard two isometric paths such that ``removing'' these paths from the graph reduces the genus of the graph by one. Let $g$-\textsf{GENUS STRING} be the class of graphs admitting a string representation on an orientable surface of genus $g$. Gaven\v{c}iak et al.~\cite{gavenciak} also used similar unfolding techniques to show that $\mathsf{c}(g$-$\textsf{GENUS STRING}) \leq 10g+15$. For this purpose, they use 10 cops to unfold a genus by guarding the closed neighborhood of two appropriate isometric paths, and then finally capturing the robber in a genus 0 string graph using 15 cops. Theorem~\ref{th:planar}, along with their unfolding techniques, gives the following immediate corollary.
\begin{corollary}
$\mathsf{c}(g$-$\textsf{GENUS STRING}) \leq 10g+13$.
\end{corollary}
For graphs having a planar representation on a surface of genus $g$, better unfolding techniques have been used. Let $G$ be a graph having genus $g$. Schroeder~\cite{schroeder} showed that $ \mathsf{c}(G) \leq \lfloor \frac{3g}{2}\rfloor + 3$. Later, Bowler et al.~\cite{bowler} improved the upper bound further and proved that $ \mathsf{c}(G) \leq \frac{4g}{3} + \frac{10}{3}$. It will be interesting to see if similar techniques can be used to improve the bounds on $\mathsf{c}(g$-$\textsf{GENUS STRING})$. Moreover, we propose the following. 

\begin{question}
Let $C$ be an isometric cycle in $G$. What is the least number of cops that can guard $N[C]$ in $G$?
\end{question}
Observe that if the answer to above question is some constant $c <10$, then we can unfold a genus by $c$ cops and therefore, we have $\mathsf{c}(g$-$\textsf{GENUS STRING}) \leq c\cdot g +13$.

Another interesting direction would be to see if the techniques used in this paper can be used to improve the cop number of unit disk graphs from 9 (Beveridge et al.~\cite{udg}) to 7. Moreover, recently de la Maza and Mohar~\cite{newMohar} characterized all 1-guardable graphs using the notion of ``wide shadow'', and used it to show that 3 cops always have a winning strategy on planar graphs even if at most 2 of the cops can move in one round. It may be interesting to study if guarding the neighborhoods of wide shadows can help improve the cop number of several graph classes.

\bibliographystyle{plainurl}
\bibliography{main}


\end{document}